\documentclass[12pt,aps]{article}
\usepackage[top=22truemm,bottom=22truemm,left=20truemm,right=20truemm]{geometry}

\usepackage{authblk}

\usepackage[utf8]{inputenc}
\usepackage{amssymb}
\usepackage{braket}
\usepackage{amsfonts}
\usepackage{amsmath}
\usepackage{bm}
\usepackage{amsthm}
\usepackage{csquotes}
\usepackage{braket}
\usepackage{enumitem}
\usepackage{xcolor}
\usepackage{hyperref}
\usepackage{pgfplots}

 \makeatletter
\def\@fnsymbol#1{\ensuremath{\ifcase#1\or \natural\or \sharp\or \flat\or
    \mathparagraph\or \|\or **\or \dagger\dagger
   \or \ddagger\ddagger \else\@ctrerr\fi}}
    \makeatother

\newcommand{\arxiv}[1]{\href{http://arxiv.org/abs/#1}{\texttt{arXiv\string:\allowbreak#1}}}

\theoremstyle{plain}
\newtheorem{thm}{Theorem}[section]

\newtheorem{prop}[thm]{Proposition}

\theoremstyle{definition}
\newtheorem{defin}[thm]{Definition}

\makeatletter

\begin{document} 

\title{Harmony for 2-Qubit Entanglement\thanks{Alberta-Thy-6-19}}


\author{Kento Osuga\thanks{osuga@ualberta.ca} }
\author{Don N. Page\thanks{profdonpage@gmail.com}}
\affil{
\textit{Theoretical Physics Institute, Department of Physics,  4-181 CCIS}, \\\textit{University of Alberta, Edmonton, Alberta T6G 2E1, Canada}
}
\date{World Music Day, 2019 June 21}
\maketitle

\begin{abstract}
In this Letter we present a new quantity that shows whether two general qubit systems are entangled, which we call \textit{harmony}. It captures the notion of separability and maximal entanglement. It is also shown that harmony is monogamous for 3-qubit states. Thus, harmony serves as a new entanglement measure. In addition, since it is written as a simple function of the density operator, it is in practice easier to compute than other previously known measures.
\end{abstract}

%
%
%


\section{Introduction}
Besides its own interesting features, quantum entanglement has become a research trend due to a notable connection between quantum aspects of black holes and quantum information theory. However, it is actually not so easy in general to tell whether two systems are entangled when the two systems are in a mixed state. In particular, the von Neumann entropy does not serve anymore as a measure of quantum entanglement for a mixed state. As of today there is no universal measure that is also practically calculable. Yet, one of the most realistic measures is the entanglement of formation:

\begin{defin}
Let $\rho$ be a density operator of a bipartite system $H_A\otimes H_B$. We consider all pure state decompositions of $\rho$ with nonnegative probabilities $p_k$ summing to unity,
\begin{equation}
\rho=\sum_{k=1}^{K}p_k\ket{\phi_k}\bra{\phi_k},\label{EoF1}
\end{equation}
where $\ket{\phi_k}$ are pure states in $H_A\otimes H_B$, and $K$ can be greater than $\dim H_A\times\dim H_B$ in general. Then, the \emph{entanglement of formation} $E(\rho)$ is defined as
\begin{equation}
E(\rho)=\text{min}\sum_{k=1}^{K}p_kE_v(\phi_k)\geq0,\label{EoF2}
\end{equation}
where the minimization is taken over all decompositions of $\rho$, and $E_v(\phi_k)=-\text{tr}(\rho_A\log\rho_A)=-\text{tr}(\rho_B\log\rho_B)$ is the von Neumann entropy of each subsystem in the bipartite pure state $\ket{\phi_k}$.
\end{defin}

By definition, two systems are separable if and only if $E(\rho)=0$. The entanglement of formation is still difficult to compute for general bipartite mixed states. If we focus on 2-qubit states, however, Wootters showed \cite{Wootters} a relatively easy way of computing the entanglement of formation:

\begin{thm}[\cite{Wootters}]\label{Wootters}
Let $\rho$ be the density operator of a 2-qubit system. For the Pauli matrix $\sigma_y$, we define the spin-flipped density operator $\tilde{\rho}$
\begin{equation}
\tilde{\rho}=(\sigma_y\otimes\sigma_y)\rho^*(\sigma_y\otimes\sigma_y),\label{rho tilde}
\end{equation}
where $\rho^*$ is the complex conjugate of $\rho$ in the standard product basis. It can be shown that  the eigenvalues of the non-Hermitian operator $\rho\tilde{\rho}$ are nonnegative. Let $\lambda_i\in\mathbb{R}_{\geq0}$ be the square roots of these eigenvalues in decreasing order\footnote{Equivalently, $\lambda_i$'s are the nonnegative eigenvalues of the positive-semidefinite Hermitian operator $R=\sqrt{\sqrt{\rho}\tilde{\rho}\sqrt{\rho}}$ in decreasing order.}. Then, the entanglement of formation $E(\rho)$ is given by
\begin{align}
E(\rho)&=h\left(\frac{1+\sqrt{1-C(\rho)^2}}{2}\right),\\
h(x)&=-x\log x-(1-x)\log(1-x),\\
C(\rho)&=\max\{0,\lambda_1-\lambda_2-\lambda_3-\lambda_4\}.\label{C}
\end{align}
Furthermore, $E(\rho)$ is a monotonically increasing function of $C(\rho)$. Accordingly, $E(\rho)=0$ if and only if $C(\rho)=0$, and $E(\rho)=\log2$ (its maximum value) if and only if $C(\rho)=1$.
\end{thm}

$C(\rho)$ is called the \textit{concurrence}. It is also known \cite{Wootters2,Bures} that $0\leq\sum_i\lambda_i\leq1$. It is not obvious that the concurrence measures the entanglement of 2-qubit systems. However, Wootter's formula given above ensures that the concurrence is a well-defined measure of entanglement.

For determining which 2-qubit states are separable $(E(\rho)=C(\rho)=0)$ and which are entangled $(E(\rho)>0)$, Wootters' formula above is a bit awkward to evaluate, because one must calculate the individual eigenvalues $\lambda_i$'s and order them to find the largest one. Therefore, we present in this Letter a simpler criterion for entanglement.

\section{Harmony}

In this section, we motivate and define \textit{harmony} $H(\rho)$ as an entanglement measure simpler to compute than the concurrence $C(\rho)$, with the criterion for entanglement of a 2-qubit state being $H(\rho)>0$. This avoids the complexity of calculating the individual eigenvalues of $R=\sqrt{\sqrt{\rho}\tilde{\rho}\sqrt{\rho}}$ and of finding the largest one.
%

\subsection{Quartet of Four Eigenvalues}
Let us consider a quantity that has no ordering complexity. The necessity of ordering the eigenvalues $\lambda_i$'s can be avoided by considering the following symmetric product $D(\rho)$ in terms of $\lambda_i$'s, which we call the \textit{disharmony}:
\begin{equation}
D(\rho)=(-\lambda_1+\lambda_2+\lambda_3+\lambda_4)(\lambda_1-\lambda_2+\lambda_3+\lambda_4)(\lambda_1+\lambda_2-\lambda_3+\lambda_4)(\lambda_1+\lambda_2+\lambda_3-\lambda_4).\label{h}
\end{equation}
With the assumption $\lambda_1\geq\lambda_2\geq\lambda_3\geq\lambda_4$, $D(\rho)$ is negative if the concurrence is positive, and $D(\rho)$ is nonnegative if the concurrence vanishes. Thus, we are motivated to define what we call the \textit{harmony} $H(\rho)$ as
\begin{equation}
H(\rho)=\text{max}\{0,-D(\rho)\}.
\end{equation}
We emphasize that the harmony is independent of the ordering of the $\lambda_i$'s. In terms of the concurrence and the smallest three eigenvalues, the harmony is
\begin{equation}
H(\rho)=C(\rho)\left(C(\rho)+2\lambda_3+2\lambda_4\right)\left(C(\rho)+2\lambda_4+2\lambda_2\right)\left(C(\rho)+2\lambda_2+2\lambda_3\right).\label{H by C}
\end{equation}
This clearly shows that $H(\rho)=0\Leftrightarrow C(\rho)=0$. Therefore, the harmony serves as an entanglement measure. 

Note that the harmony is not purely a function of the concurrence. Putting it another way, for a fixed value of the concurrence $C$, there is a range from $H_{\text{min}}(C)$ to $H_{\text{max}}(C)$  that $H(\rho)$ can vary. Since all $\lambda_i$'s are nonnegative, and since their sum does not exceed unity, if there are no other constraints on them, it immediately follows that $H_{\text{min}}(C)=C^4$ when $\lambda_1=C$ and $\lambda_2=\lambda_3=\lambda_4=0$. Also, it is clear from Eq. \eqref{H by C} that we have $\sum_i\lambda_i=1$ when the harmony is maximized at fixed $C$, so that the sum of the last three factors in Eq. \eqref{H by C} is $2+C$. In this case the product of the three factors is maximized when each factor is equal, so each is $(2+C)/3$. This gives
\begin{equation}
H_{\text{max}}(C)=\frac{C(2+C)^3}{27},\label{max}
\end{equation}
where $H(\rho)=H_{\text{max}}(C)$ if and only if
\begin{equation}
\lambda_1=\frac{1+C}2,\;\;\;\;\lambda_2=\lambda_3=\lambda_4=\frac{1-C}{6}.
\end{equation}
In particular, since the concurrence goes from 0 to 1, Eq. \eqref{max} implies that $H(\rho)\leq1$, where the inequality is saturated if and only if $C(\rho)=1$. Therefore, the harmony is maximized for the maximally entangled states, which maximize the concurrence.


\subsection{Practice}

We have shown how the harmony $H(\rho)$ is by construction independent of the ordering of $\lambda_i$'s, which overcomes one of the two computational difficulties mentioned above for calculating the concurrence $C(\rho)$. However, it is also somewhat awkward to calculate the individual eigenvalues. Therefore, we now show that the harmony can be written more directly as a simple function of the density operator $\rho$.

Let us expand the disharmony $D(\rho)$ in Eq. \eqref{h} and combine terms into the following way:
\begin{equation}
D(\rho)=-2\sum_{i=1}^4\lambda_i^4+\left(\sum_{i=1}^4\lambda_i^2\right)^2+8\lambda_1\lambda_2\lambda_3\lambda_4.
\end{equation}
Then, notice that these three terms can be written as
\begin{equation}
D(\rho)=-2\;{\rm tr}\left[(\rho\tilde{\rho})^2\right]+\left[{\rm tr}(\rho\tilde{\rho})\right]^2+8\det\rho,
\end{equation}
which, since $\rho^* = \rho^T$, is an analytic function (a multivariate 4th-order polynomial) in the components of the density matrix $\rho$. Each term in the above expression can be computed by a simple operation, simpler than obtaining the eigenvalues $\lambda_i$'s. As a consequence, the harmony can be thought of as a practically more easily computable entanglement measure for 2-qubit states than the concurrence or entanglement of formation.

In summary, we give a simple definition and a theorem for the harmony of a 2-qubit quantum state.

\begin{defin}
Let $\rho$ be the density operator of a 2-qubit system, and let $\tilde{\rho}$ be the density operator given by Eq. \eqref{rho tilde}. Then the \emph{harmony} $H(\rho)$ of the state $\rho$ is defined as
\begin{equation}
H(\rho)={\rm max}\{0,2\;{\rm tr}\left[(\rho\tilde{\rho})^2\right]-\left[{\rm tr}(\rho\tilde{\rho})\right]^2-8\det\rho\}.
\end{equation}
\end{defin}

\begin{thm}\label{thm}
Let $\rho$ be the density operator of a 2-qubit system and $H(\rho)$ be its harmony. Then the following statements are true:
\begin{itemize}
\item The two qubits are separable if and only if $H(\rho)=0$.
\item The two qubits are maximally entangled if and only if $H(\rho)=1$.
\end{itemize}
\end{thm}

Wootters' formula (Theorem~\ref{Wootters}) implies that the states with $H(\rho)=1$ also maximize the concurrence $C(\rho)$ and entanglement of formation $E(\rho)$, and vice versa. It is shown in \cite{Max} that every maximally entangled state in a $d\times d$ bipartite state is indeed pure. Then it follows that most other entanglement measures, including distillable entanglement \cite{DE}, entanglement cost \cite{EC}, relative entropy of entanglement \cite{RE}, squashed entanglement \cite{SQ1,SQ2}, and negativity \cite{N}, are maximized by the states that gives $H(\rho)=1$.

\subsection{Harmony is not Monotone}

It can be explicitly shown that harmony is invariant under local unitary transformations $\rho\mapsto U_A\otimes U_B\cdot\rho\cdot U_A^{\dagger}\otimes U_B^{\dagger}$. On the other hand, harmony is not a convex function over the set of density operators, as one can see from the following three density operators $\rho_{\pm},\rho$ with $0\leq x\leq1$, the first two of which are pure:
\begin{equation}
\rho_{\pm}=\frac{1\pm x}2\ket{00}\bra{00}+\frac{\sqrt{1-x^2}}2(\ket{00}\bra{11}+\ket{11}\bra{00})+\frac{1\mp x}2\ket{11}\bra{11},
\end{equation}
\begin{equation}
\rho=\frac12\rho_++\frac12\rho_-.
\end{equation}
Then, we have
\begin{equation}
H(\rho)=1-x^2\geq(1-x^2)^2=\frac12H(\rho_+)+\frac12H(\rho_-).
\end{equation}
Thus, harmony is not an entanglement monotone \cite{Monotone}. In exchange for the loss of monotonicity, harmony has the advantage that it can be evaluated by a simpler computation than the negativity or the concurrence, which require the ordered set of eigenvalues of the corresponding $4\times4$ matrices.

%
%
%

\subsection{Performance for Pure States}

Even though our purpose of introducing the harmony is to define an entanglement measure applicable to mixed states, let us consider the harmony for pure 2-qubit states. Setting the basis of a 2-qubit system as $\{\ket{00},\ket{01},\ket{10},\ket{11}\}$, the $\sigma_y\otimes\sigma_y$ spin-flip operation transforms the basis by
\begin{equation}
\ket{00}\mapsto-\ket{11},\;\;\;\;\ket{11}\mapsto-\ket{00},\;\;\;\;\ket{01}\mapsto\ket{10},\;\;\;\;\ket{10}\mapsto\ket{01}.
\end{equation}
For a general pure state
\begin{equation}
\ket{\psi}=a\ket{00}+b\ket{11}+c\ket{01}+d\ket{10}\;\;\;\;\text{with}\;\;\;\;|a|^2+|b|^2+|c|^2+|d|^2=1,
\end{equation}
one then gets
\begin{equation}
\ket{\tilde{\psi}}:=\sigma_y\otimes\sigma_y\ket{\psi}^*=-b^*\ket{00}-a^*\ket{11}+d^*\ket{01}+c^*\ket{10},
\end{equation}
\begin{equation}
H(\rho)=|\braket{\tilde{\psi}|\psi}|^4.
\end{equation}
Note that for any pure state $\rho=\ket{\psi}\bra{\psi}$, the purity $\gamma(\rho_A)=\text{tr}(\rho_A^2)$ of the reduced density operator $\rho_A=\text{tr}_B(\rho)$ and the concurrence $C(\rho)$ can be written as functions purely of the harmony $H(\rho)$:
\begin{align}
\gamma(\rho_A)&=1-\frac{1}2|\braket{\tilde{\psi}|\psi}|^2=1-\frac{1}2\sqrt{H(\rho)},\\
C(\rho)&=\sqrt{2[1-\gamma(\rho_A)]}=\sqrt{4\det\rho_A}=\left[H(\rho)\right]^{\frac14}.\label{C pure}
\end{align}
Therefore, for a pure 2-qubit state, the harmony is simply the fourth power of the concurrence, as one can also see from Eq. \eqref{H by C}, since then $\lambda_2=\lambda_3=\lambda_4=0$.

\section{Harmony of Monogamy}\label{sec:Monogamy}

We now present another piece of evidence that the harmony is a good measure of entanglement. To be more explicit, we prove that the harmony is monogamous for any 3-qubit pure state. Since $H(\rho)$ is not a function of $C(\rho)$, the argument given in \cite{Monogamy} for any monotonically increasing function $\Gamma(C)$ is not applicable, yet the discussion below is inspired by \cite{Monogamy}.

Let us first give a few facts to illustrate monogamy. Given a 3-qubit system $X,Y,Z$, we pick a pure state $\rho$. Then \cite{Wootters} there are at most two nonzero eigenvalues $\lambda_i$'s associated with $\rho_{YZ}\tilde{\rho}_{YZ}$. This implies that even though system $YZ$ in principle has 4 dimensions, the pair $(Y,Z)$ can be interpreted as having an effectively \textit{two-dimensional} Hilbert space. Moreover, we can consider the harmony $H_{X(YZ)}$ between $X$ and the pair $(Y,Z)$. Knowing from Eq. \eqref{C pure} that $C(\rho_{XW})=\sqrt{4\det\rho_X}$ for any pure state $\rho_{XW}$ in any 2-qubit system $X$ and $W$, we have
\begin{equation}
H_{X(YZ)}=(4\det\rho_X)^2=\left(C_{X(YZ)}\right)^4.\label{H3}
\end{equation}
The monogamy is an inequality among $H_{XY}$, $H_{XZ}$, and $H_{X(YZ)}$, where $H_{XY}=H(\rho_{XY})$:

\begin{prop}
Let $\rho$ be a pure state of a 3-qubit system $X,Y,Z$. Then, we have
\begin{equation}
H_{XY}+H_{XZ}\leq H_{X(YZ)}.\label{mono}
\end{equation}
\end{prop}
\begin{proof}
The proof closely follows \cite{Monogamy}. Since only $\lambda_1$ and $\lambda_2$ are nonzero for $\rho_{XY}$, the harmony is
\begin{equation}
H_{XY}=(\lambda_1-\lambda_2)^2(\lambda_1+\lambda_2)^2=(\lambda_1^2+\lambda_2^2)^2-4\lambda_1^2\lambda_2^2\leq\left(\text{tr}(\rho_{XY}\tilde{\rho}_{XY})\right)^2.
\end{equation}
Then, it follows that
\begin{align}
H_{XY}+H_{XZ}&\leq\left(\text{tr}(\rho_{XY}\tilde{\rho}_{XY})\right)^2+\left(\text{tr}(\rho_{XZ}\tilde{\rho}_{XZ})\right)^2\nonumber\\
&\leq\left(\text{tr}(\rho_{XY}\tilde{\rho}_{XY})+\text{tr}(\rho_{XZ}\tilde{\rho}_{XZ})\right)^2=\left(C_{X(YZ)}\right)^4,
\end{align}
where the last equality uses the result of \cite{Monogamy}. Finally, Eq. \eqref{H3} gives the monogamy relation \eqref{mono}.
\end{proof}

Due to the non-convexity of the harmony, we are no longer able to apply the same argument that \cite{Monogamy} used to extend the monogamy of the concurrence for mixed 3-qubit states. However, we can still explore a monogamous relation for mixed states which is not so constrained as the one for the concurrence discussed in \cite{Monogamy}.

Let $\rho$ be the density operator for a mixed 3-qubit state, and then \cite{Monogamy} derived the following monogamous relation for the squares of the concurrences:
\begin{equation}
C^2_{XY}+C^2_{XZ}\leq (C^2)_{X(YZ)}^{\text{min}},\label{monoC}
\end{equation}
\begin{equation}
(C^2)_{X(YZ)}^{\text{min}}=\text{min}\sum_ip_iC^2_{X(YZ)i},\label{Cmin}
\end{equation}
where minimization is taken over all pure state decompositions, hence each $C^2_{X(YZ)i}$ still makes sense. It can be easily shown by Eq. \eqref{max} that $H(\rho)\leq C(\rho)$ for any density operator $\rho$. Then, we can rewrite the monogamous relation \eqref{monoC} for the harmony by using Eq. \eqref{C pure} as follows:

\begin{prop}
Let $\rho$ be a mixed state of a 3-qubit system $X,Y,Z$. Then, we have
\begin{equation}
H^2_{XY}+H^2_{XZ}\leq (H^{\frac12})_{X(YZ)}^{\text{min}},\label{mono2}
\end{equation}
where $(H^{\frac12})_{X(YZ)}^{\text{min}}$ is defined similarly to Eq. \eqref{Cmin}. 
\end{prop}

In particular, we know that $(H^{\frac12})_{X(YZ)}^{\text{min}}\leq1$, which implies that $H_{XY}=1$ if and only if $H_{XZ}=0$ and vice versa. Therefore, even though the upper bound in the inequality \eqref{mono2} is not as constrained as the one in \eqref{mono} or the one for the concurrences \eqref{monoC}, the notion of monogamy still stands for mixed 3-qubit states as expected. More constrained monogamous relations remain to be seen.

\section{Conclusion}
We have introduced a new entanglement measure for 2-qubit states that we call harmony. We illustrated that it provides a criterion for separable states as well as for maximally entangled states. We also showed that the harmony is monogamous for 3-qubit states. Harmony behaves very similar to concurrence, but it is worth emphasizing that it is not a function purely of concurrence. In exchange for losing monotonicity, harmony is written as a simpler function of the density operator $\rho$,
\begin{equation}
H(\rho)={\rm max}\{0,2\;{\rm tr}\left[(\rho\tilde{\rho})^2\right]-\left[{\rm tr}(\rho\tilde{\rho})\right]^2-8\det\rho\},
\end{equation}
where $\tilde{\rho}=(\sigma_y\otimes\sigma_y)\rho^*(\sigma_y\otimes\sigma_y)$. Therefore, for a given density operator $\rho$ the harmony should be easier to compute.

As the harmony is given by a simple expression, it is an interesting question whether its definition can be extended beyond 2-qubit states.

\section*{Acknowledgment}
This research was supported by the Natural Sciences and Engineering Research Council of Canada.


\end{document}